\documentclass[journal]{IEEEtran}

\usepackage{amsfonts}
\usepackage{amsmath,graphicx}
\usepackage{amssymb}
\usepackage{epsfig}
\usepackage{cite}
\usepackage{color, soul}
\usepackage{bm}
\usepackage{ifpdf}
\usepackage{amsthm}
\usepackage{array}
\usepackage{multirow}{\large }
\usepackage{booktabs}
\usepackage{subfig}

\newtheorem{definition}{Definition}

\newtheorem{theorem}{Theorem}%\newtheorem{theorem}{Theorem}[section]

\newtheorem{lemma}{Lemma} %\newtheorem{lemma}[theorem]{Lemma}

\newtheorem{example}{Example}
\newtheorem{scenario}{Scenario}

\begin{document}

\title{Key Pre-Distributions From \\ Graph-Based Block Designs}

\author{Jie~Ding, Abdelmadjid~Bouabdallah, 
        and~Vahid~Tarokh
        \thanks{This work was carried out in the framework of the Labex MS2T, which was funded by the French Government, through the program ``Investments for the future'' managed by the National Agency for Research (Reference ANR-11-IDEX-0004-02).}
       \thanks{J. Ding and V. Tarokh are with the School of Engineering and Applied Sciences, 
		Harvard University, Cambridge, MA 02138, USA. E-mail: jieding@g.harvard.edu , vahid@seas.harvard.edu .}
        \thanks{A. Bouabdallah is with the Sorbonne universit$\acute{\textrm{e}}$s, Universit$\acute{\textrm{e}}$ de technologie de Compi$\grave{\textrm{e}}$gne, CNRS, Heudiasyc UMR 7253, CS 60 319, 60203 Compi$\grave{\textrm{e}}$gne cedex, France. E-mail: bouabdal@hds.utc.fr .}
%     \thanks{A. Bouabdallah is with the Universit$\acute{\textrm{e}}$ de Technologie de Compi$\grave{\textrm{e}}$gne, Laboratoire HeuDiaSyc, UMR CNRS 7253, Compi$\grave{\textrm{e}}$gne, France. E-mail: bouabdal@hds.utc.fr .}
      }
%\thanks{J. Ding and V. Tarokh are with the School of Engineering and Applied Sciences, 
%Harvard University, Cambridge,
%MA, 02138 USA e-mail: jieding@g.harvard.edu, vahid@seas.harvard.edu}% <-this % stops a space
%%\thanks{Y. Rachlin is with the Lincoln Laboratory, Massachusetts Institute of Technology, Lexington, MA 02420 USA e-mail: yaron.rachlin@ll.mit.edu}

%\markboth{IEEE Transactions on Wireless Communications}%IEEE TRANSACTIONS ON WIRELESS COMMUNICATIONS
%{Shell \MakeLowercase{\textit{et al.}}: Bare Demo of IEEEtran.cls for Journals}



\maketitle

\begin{abstract}

%We propose methods to incorporate prior knowledge of network characteristics and constraints into the design of key pre-distribution schemes. This provides better security and connectivity while requiring less resources compared with many existing solutions.
%Our methods are based on casting the prior information as a graph, and associated graph-based approach to key pre-distribution schemes.  We apply a class of quasi-symmetric designs referred here to as $g$-designs as a design tool. These produce key pre-distribution schemes that significantly improve upon the existing constructions based on unital designs, in terms of network scalability, resiliency, storage overhead, etc. %Motivated by these potentials, we initiate the study of $g$-designs, construct some examples, and point out open problems for future research. 
%%We give some examples, and point out open problems for future research.

With the development of wireless communication technologies which considerably contributed to the development of wireless sensor networks (WSN), we have witnessed an ever-increasing WSN based applications which induced a host of research activities in both academia and industry. Since most of the target WSN applications are very sensitive, security issue is one of the major challenges in the deployment of WSN. One of the important building blocks in securing WSN is key management. Traditional key management solutions developed for other networks are not suitable for WSN since WSN networks are resource (e.g. memory, computation, energy) limited. Key pre-distribution algorithms have recently evolved as efficient alternatives of key management in these networks. In the key pre-distribution systems, secure communication is achieved between a pair of nodes either by the existence of a key allowing for direct communication or by a chain of keys forming a key-path between the pair. 

In this paper, we propose methods which bring prior knowledge of network characteristics and application constraints into the design of key pre-distribution schemes, in order to provide better security and connectivity while requiring less resources. Our methods are based on casting the prior information as a graph.  Motivated by this idea, we also propose a class of quasi-symmetric designs 
referred here to as g-designs. %This provides better security and connectivity while requiring less resources. 
These produce key pre-distribution schemes that significantly improve upon the existing constructions based on unital designs. We give some examples, and point out open problems for future research.

\end{abstract}

\begin{IEEEkeywords}
Balanced incomplete block design, quasi-symmetric design, key pre-distribution,
sensor networks, graphs. 
\end{IEEEkeywords}

% \fi
\IEEEpeerreviewmaketitle

%============================================================ SECTION =====================
\section{Introduction} \label{introSection}

\setcounter{equation}{0}
\renewcommand{\theequation}{\arabic{equation}}


%\IEEEPARstart{T}{his} demo file is intended to serve as a ``starter file''

 \IEEEPARstart{W}{ireless} sensor networks (WSN) typically consist of  a large number of sensor nodes with limited memory and power.  These networks are used in both military and civilian applications. In military applications, sensor nodes may be deployed for example in battlefield surveillance, while in environmental applications, distributed sensors could monitor physical or environmental conditions such as temperature, sound, pressure, etc. and cooperatively pass their data through the network to a main location \cite{yick2008wireless,akyildiz2002wireless}. 
 Because of the sensitivity of most WSN applications, security issue is one of the major challenges in the deployment of WSN.

Security is an essential question for many sensor network applications, especially for military applications. Providing security to small sensor nodes is challenging because of the resources limitations of sensor nodes in terms of storage, computations, communications, and energy. One of the important building blocks for the development of security solutions in WSN is key management. The key management scheme design is more complicated due to the characteristics of the WSN such as: (1) The vulnerability of nodes to physical attack, where the deployment in a hostile area makes the attacker capable to simply compromise any node and to reveal its security materials (e.g. keys, functions, ...); 
%(2) The unknown network topology prior to deployment, where in most applications the sensors are randomly deployed. So, a key management scheme should not assume a priori knowledge of the network topology; 
(2) The nature of wireless communication, where the radio links are insecure and an attacker can eavesdrop on the radio transmissions, inject bits in the channel, and replay previously overheard packets; 
(3)  The density and the large size of the network which make the control of all nodes very hard.

For security or privacy reasons, it is often critical to build encrypted communications between two sensor nodes using a common secret key. 
Key pre-distribution scheme (KPS) is a classical way to set up secret keys among sensor nodes before the deployment phase.
Compared with online key exchange protocols, key pre-distribution is more attractive for networks
consisting of large number of nodes with limited communication/computation resources \cite{RKP}. 

Over the last decade, a host of research on key pre-distribution issue for WSN have been conducted and many solutions have been proposed in literature. % [?] [?] [?]. 
Existing KPS fall into two categories: probabilistic  and deterministic schemes. 
In probabilistic schemes,  a direct connection between each two nodes is established with certain probability, i.e. probability that these two nodes share a common key.
In deterministic schemes, however, each pair of nodes are known to be directly connected or not. 

Eschenauer and Gligor \cite{RKP} proposed a random key pre-distribution (RKP) scheme. Later on, there have been some  improvements on RKP, e.g. 
improvements on its resilience by increasing the ``intersection threshold''
(the least number of common keys for two nodes to establish a direct connection) 
\cite{Chan}, 
or improvements on its resilience as well as higher probability of sharing keys by 
using deployment knowledge \cite{Du,Group_deployment,ito2005key,liu2008group,wei2005product}.
Schemes called ``multiple key spaces'' are proposed that combine different KPS to achieve better performance \cite{du2005pairwise,lee2005deterministic,liu2005establishing}.

A simple deterministic scheme assigns a distinct key to each link,  and $b-1$ pairwise keys to each node, where $b$ is the number of nodes.
Choi et al. \cite{Choi} proposed an improved method where each node only needs to store $(b+1)/2$ keys. However, these methods may not be easily scalable.
Camtepe et al. \cite{camtepe2004combinatorial,pioneer}  proposed a novel method that uses block designs for key pre-distribution. 
They proposed a deterministic key pre-distribution scheme that maps a symmetric balanced incomplete block design (SBIBD) or  
generalized quadrangles (GQ) to key pre-distribution. 
%Usually, the set of keys assigned to a sensor node is called its ``key ring'', and the set of all keys is called ``key pool''.
%Specifically, the SBIBD-based KPS allows to construct $m^2 +m+1$ key rings (the set of keys assigned to a sensor node) from a key pool (the set of all keys) of $m^2 +m+1$ keys, 
%such that each key ring has $m+1$ keys and each pair of key rings share exactly one common key. 

%---- The following content seems redundant, so I deleted 
%\subsection{Key Pre-Distribution Scheme Using Block Design}
%
%Though the field of block design and key predistribution have been extensively studied. The application of block design to key predistribution is relatively new.
%Camtepe and Yener initially used combinatorial design theory in key predistribution problems \cite{pioneer}.  The SBIBD-based KPS, as is introduced in the previous section, is a deterministic 
%scheme that guarantees the connection of each pair of nodes. 

There are some other works that use design theory to construct effective KPS.
Lee et al. \cite{lee2005combinatorial} introduced the common intersection designs to KPS.
Chakrabarti et al. \cite{chakrabarti2006key} used transversal designs and merging block techniques. 
Ruj et al. \cite{PBIBD1} proposed a KPS that is based on partially balanced incomplete block designs (PBIBD).
They also used triple system to design a probabilistic KPS \cite{BD4KPS2}. 
Later on, Bose et al. %proposed a construction method that achieves better performance by combining several PBIBDs \cite{PBIBD2}.
\cite{PBIBD2} proposed an improved construction that combines several PBIBDs.
Bechkit et al. \cite{unital} proposed the unital-based key pre-distribution scheme (UKP), a deterministic scheme that improves the scalability of a network. More detailed surveys could be found in \cite{hwang2004revisiting,camtepe2005key,lee2008construction}.
%The UKP applies unital designs %(a special case of non-symmetric BIBD) 
%to KPS. It gives $m^2(m^2 -m+1)$ key rings from a key pool of $m^3+1$ keys, such that each key ring has $m+1$ keys.

\subsection{Motivations for Graph-Based KPS} \label{sec_motivation2GKPS}

Usually,  the main goal of KPS is to seek a design solution that provides more connectivity coverage while requiring less memory  (or using as few keys as possible). 
Existing designs are often evaluated under criteria such as 
network connectivity, average path length, network resiliency, storage overhead, and network scalability.

Our contributions are motivated by the following observations:
\begin{itemize}

\item {\bf Observation 1}: In many applications, there are several intended deployment locations, and typically a number of sensor nodes are deployed in each of these locations.
%\footnote{For example, a group of sensors are scattered around a desired location. Readers may refer to \cite{Du,Group_deployment} for more concrete models.}    
In hierarchical scenarios, each group of sensors placed in a location must pass their data to sensor nodes of higher ranks. This means that for a hierarchical sensor network, nodes in the same group need more connectivity than those across groups (Figure \ref{Obs1}). This can be used to reduce the number of required keys.% while improving on the security and performance.

\item {\bf Observation 2}: In some scenarios, a group of sensor nodes are naturally set in a master-slave architecture. For instance, the commander of a military needs more communications with his lower rank staff.
This means that one or some nodes are in charge of collecting data/sending command signals to the remaining nodes  (Figure \ref{Obs2}). So there are some important connections that we need to establish with higher efficiency and security.

\item {\bf Observation 3}:  Two sensor nodes can communicate with each other only in a certain distance referred to as the radio frequency (RF) range (Figure \ref{Obs3}).  If it is known that certain pairs never connect directly due to being outside of each other's radio frequency range, then it would be a waste of resources to assign common keys to them.

\item {\bf Observation 4}: If it is known in advance that certain connections are more likely to be eavesdropped, the corresponding nodes should not share common keys  (Figure \ref{Obs4}). %This is especially important in military applications.
\end{itemize}

\begin{figure}[h]%[!t]
\centering
\subfloat[]{\includegraphics[width=2.5in]{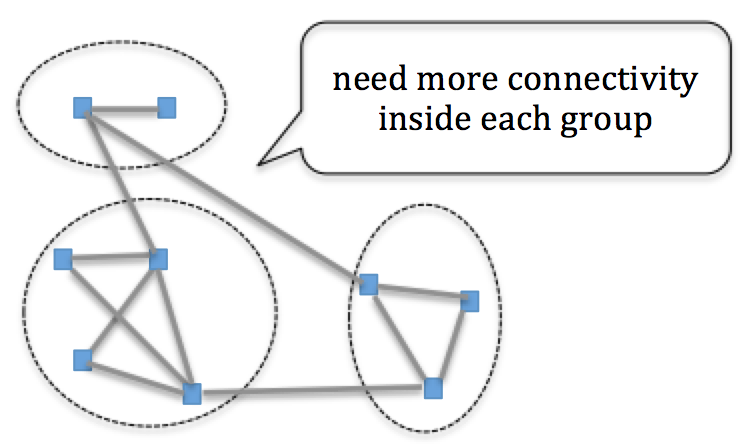}%
\label{Obs1}}
\hfil
\subfloat[]{\includegraphics[width=2.1in]{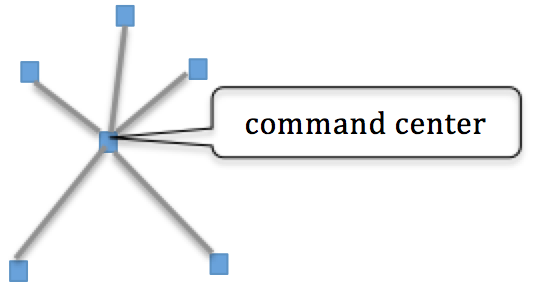}%
\label{Obs2}}
\\
\subfloat[]{\includegraphics[width=2in]{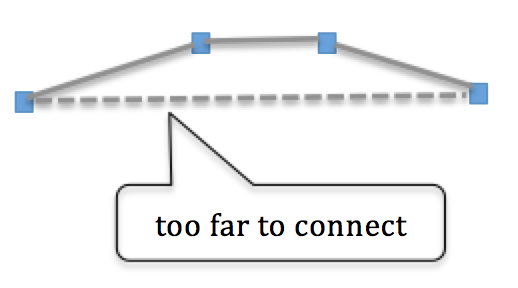}%
\label{Obs3}}
\hfil
\subfloat[]{\includegraphics[width=2.5in]{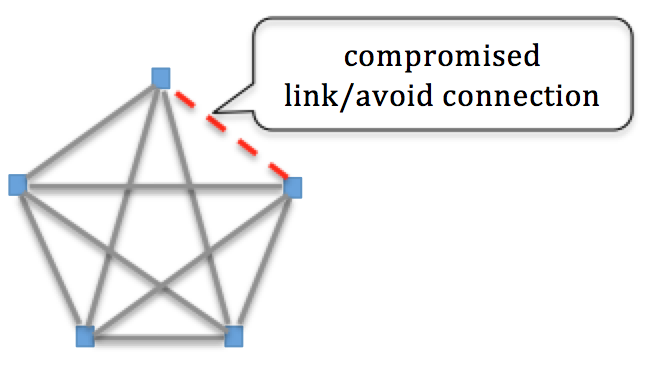}%
\label{Obs4}}
\caption{Observations}
\label{Observations}
\end{figure}

These observations motivate the use of prior information in designing improved KPS schemes. 

%-- The following part is old and redundant 
%---
%Graph-based block design and graph-based key predistribution
%---
%
%Graph-based block design is a new concept that has not been put forward. A possible reason is that 
%given a valid parameter set, it is generally hard to prove the existence or to construct a BIBD. Thus a block design that is based on a given graph is even harder. 
%
%Graph-based key predistribution is  also new. Some previous work may look similar, but is actually different in concept. For example, 
%Du et al. used deployment knowledge to optimize KPS \cite{Du}.
%But the concept of 'deployment' is in a geographic sense and is modeled as grid. Furthermore, the aim of using deployment knowledge is to 
%improve the overall performance but not to satisfy a desired distribution pattern that is modeled as a graph. 
%
%This paper focuses on the graph-based block design and its application in KPS --- graph-based KPS. 

\subsection{Contribution and Organization of This Work} \label{sec_org}

We propose graph-based key pre-distributions (block designs) that 
incorporate prior knowledge of network characteristics and constraints. 
This provides better security and connectivity while requiring less resources once properly used.
We elaborate on two practical scenarios, and explain why the graph-based design is preferable or even required. 
Some previous work that used deployment knowledge to optimize KPS may look similar, but is different in concept. 
For example, the scope of ``deployment'' was usually in a geographic sense. However, the model that communication should not pass through some links of potential danger was rarely studied, up to our knowledge.
A general framework that casts the prior information as graphs is the main motivation of this work.

We first briefly review the basics of block design theory in Section \ref{sec_BIBDIntro}.
Especially, we propose $g$-designs, a class of quasi-symmetric designs, and initiate the study of the applications of them to KPS.
We propose the concept of graph-based KPS in Section \ref{sec_graph_design}. 
In Section \ref{sec_example_can}, we demonstrate the improvements provided by graph-based KPS and 
$g$-designs in a specific scenario.
In Section \ref{sec_example_should} we study another scenario. We further provide an algorithmic framework (referred to MAR) for KPS design for the second scenario in Section \ref{sectionMAR}.
Finally, we make our conclusions in Section \ref{sec_conclusion}.% and discuss directions for future research.

%============================================================ SECTION =====================
\section{Background} \label{sec_BIBDIntro}

%======== subsection ===
\subsection{Block Design Theory}

Block design theory deals with the properties, existence, and construction of systems of finite sets whose intersections have specified numerical properties. 
A block design (BD) is a set system ($V, \mathfrak{B}$). Each element in $V$ is called a point (or treatment), and each 
element in $\mathfrak{B}$ is called a block.

\begin{definition} \label{def_parSet} 
A \textbf{Balanced Incomplete Block Design (BIBD)} with parameters $\lambda$, $k$, $r$, $v$, and $b$ is a block design in which $v$ points are arranged in $b$ blocks, such that each block contains $k$
points,  each point appears in $r$ blocks, and each pair of points appear in exactly $\lambda$ blocks.
It is denoted by $(\lambda,k,r,v,b)$-BIBD. It may also be denoted by $(\lambda, k, v)$-BIBD, as $r$ and $b$ are given by \cite{BD}
\begin{equation} \label{ddl}
r = \frac{\lambda (v-1)}{k-1} , \quad 
b =  \frac{\lambda (v-1) v }{(k-1) k}.
\end{equation}
\end{definition}

\begin{definition} \label{def_sDesign}
A BIBD is a \textbf{$g$-design} if any two blocks intersect in either zero or a fixed number of points $g>0$. 
\end{definition}

A quasi-symmetric design is a BIBD with two possible intersection numbers for any pair of blocks. A $g$-design clearly belongs to the class of quasi-symmetric designs when one intersection number equals to zero. But the term $g$-design is defined here for convenience.  
Some properties and constructions of $g$-designs have been studied in \cite{shrikhande1986survey,qs_0y,qs_02,qs,colbourn2010handbook}. 

Clearly, any $(\lambda, k, r, v, b)$-BIBD with $\lambda=1$ (also referred to as a Steiner system) is a $g$-design with $g=1$. 

Also, any unital design is a $( 1, m+1, m^2, m^3+1, m^2(m^2-m+1) )$-BIBD for some $m$, and is also a $g$-design.

\begin{definition} \label{def_graph}
Let $G$ be a graph. Let $V(G)$ and $E( G)$ be respectively the set of nodes and edges of $G$. 
$G$ is  \textbf{regular} with \textbf{degree} $d$ if every node of $G$ %is directly connected with $d$ other nodes. 
has incidence degree $d$. A clique in $G$ is a subset of its vertices such that every two vertices in the subset are connected by an edge; in other words, it is a subgraph of $G$ and is complete.
%The instance matrix of $ G$ is denoted by $\bm G$, i.e. the $(m,n)$th element of $\bm G$ is $1$ if and only if the $m$th  
%and $n$th nodes are directly cconnected, and $0$ otherwise. 
 % $\bm A_G $. For simplicity, we use $\bm G$ instead of $\bm A_G$ when there is no ambiguity.
\end{definition}

\begin{definition} \label{def_srg}
A \textbf{strongly regular graph (SRG)} with parameter set $(b,d,t,u)$ is defined as a regular graph
of size $b$ and degree $d$,
%\footnote{ In other words, the graph contains $b$ nodes, each of which connects with $d$ other ones.}
such that
every two adjacent nodes have $t$ common neighbors, 
and every two non-adjacent nodes have $u$ common neighbors.
The SRG is denoted by srg$(b,d,t,u)$.
\end{definition}

%======== subsection ===

\subsection{Some Properties of $g$-Designs} \label{sec_motivation2sDesign}

%There are several aspects that motivates the study of $g$-design. 

%\subsubsection{As a special graph-based design}

%For any $g$-design, it can be regarded as the graph-based design with $ G_T = ( G_D, \bar{ G}_B, \emptyset )$
%being the target graph, where $ G_D$ is its design graph and $ \bar{ G}_B$ is the complement of $ G_D$. 
%This graph-based design is so special that any pair of connected blocks share a constant number of keys.
%
%\subsubsection{As a generalization of unital designs}
%

%SBIBD is a $(1, m+1, m+1, m^2+m+1, m^2+m+1)$-BIBD

The following results are helpful for the future analysis.

\begin{theorem} \label{bridge_s2}
If there exists a $(\lambda, k, v)-$BIBD which is also a $g$-design with $g=2$, 
then there exists a regular graph $ G$ of size $b$ such that 
its edge set $E( G)$ is a disjoint union of $v(v-1)/2$ subsets, where each subset 
 forms a clique of size $\lambda$. 
 Also, $b$ satisfies the equation %there exists a positive integer $k$ such that   
\begin{equation} \label{bridge2_par}
k(k-1) = \frac{\lambda v (v-1)}{b}.
\end{equation}
\end{theorem}

\begin{proof}
The proof is given in Appendix \ref{proof_bridge_s2}, as it uses a definition to be introduced in Section \ref{sec_graph_design}.
\end{proof}

Theorem \ref{bridge_s2} provides a necessary condition for the $g=2$ case.
The following result shows the equivalence between $g$-designs with $g=1$ and a class of graphs.
%The result generalizes relevant ones in Example \ref{example_cute}.

\begin{theorem} \label{bridge_s1}
The existence of a ($\lambda=1, k, v)$-BIBD is equivalent to the existence of a regular graph $ G$
of size $b$
such that  
its edge set $E( G)$ is a disjoint union of $v$ subsets each of 
which forms a clique of size $r$, where  
\begin{equation} \label{bridge1_par}
\frac{rv}{b} = \frac{v-1}{r} + 1 \in \mathbb{N}, %b=\frac{rv}{k}, \quad r=\frac{v-1}{k-1}, 
\end{equation}
 $\mathbb{N}$ is the set of natural numbers. 
\end{theorem}

\begin{proof}
The proof is given in Appendix \ref{proof_bridge_s1}, as it uses a technique to be introduced in Section \ref{sectionMAR}.
\end{proof}

\begin{lemma} (\cite{sane1987quasi}, Lemma 2.1)

If a $(\lambda, k, v)$-BIBD is a $g$-design, then its design graph is a strongly regular graph, denoted by srg$(b,d,t,u)$,
%and 
%\begin{align}
%r = \frac{ \lambda(v-1)}{k-1} , \quad b =  \frac{ \lambda (v-1) v }{(k-1) k} \label{d4} .
%\end{align}
When $\lambda  = 1$, we have
\begin{align}
d &= \frac{v-k}{k-1}k, \quad t =\frac{v-1}{k-1}-2+(k-1)^2, \quad u = k^2. \label{newd37} 
\end{align}
\end{lemma}

%============================================================ SECTION =====================
\section{Graph-Based KPS and Evaluation Metrics} \label{sec_graph_design}

Block designs are intimately related to key pre-distribution schemes. In a KPS system, each sensor node is assigned a set of keys, called ``key ring''.  Let a key correspond to a point, and a key ring correspond to a block.
%\footnote{It makes no essential difference if key corresponds to a block, and a key ring corresponds to a key. But sometimes it may be more convenient to switch, e.g. when PBIBD is applied to KPS \cite{PBIBD1}. }
For example, A unital design-based KPS gives $m^2(m^2 -m+1)$ key rings from a key pool of $m^3+1$ keys, such that each key ring contains $m+1$ keys.

In sensor network applications, if two key rings share at least one common key, the corresponding two nodes can be directly connected to each other. The direct connections form a graph, whose nodes correspond to the sensor nodes, and edges correspond to direct connections.  This graph is referred to as the ``design graph''.  In mathematical terms:

 \begin{definition} \label{def_targetGraph} 
 A \textbf{design graph} for a specific block design is a graph $ G_D$ whose nodes $V( G_D)$ correspond 
 to the blocks. Two nodes are connected in $ G_D$ if and only if the corresponding two blocks share at least one point. 
 \end{definition}
 
 The prior structural information of a network (discussed) may also be modeled as a graph:
 \begin{definition} \label{def_targetGraph} 
 % with vertex set $V(\bm G)$ and edge set $E(\bm G)$
A \textbf{target graph} for a specific WSN is a triplet of graphs $ G_T = ( G_T^c,  G_T^u,  G_T^r)$ that satisfies
 \begin{enumerate}
 \item each node of $ G_T^c$, $ G_T^u$, or $ G_T^r$ corresponds to a node in the WSN,
 \item two nodes are connected in $ G_T^c$ if and only if the corresponding nodes in the WSN must directly communicate,
 \item two nodes are connected in $ G_T^u$ if and only if the corresponding nodes in the WSN are required not to directly communicate, and
 \item for any pair of nodes not covered by the above two cases, they are not connected in $ G_T^r$ if and only if the corresponding nodes in the WSN %are known not to need to directly communicate. 
 may communicate via a path but not necessarily communicate directly.
 \end{enumerate}
 \end{definition}
 
We note that design graphs and target graphs are undirected and unweighted.
Also in the classical case,  $ G_T^c,  G_T^u$ are null graphs (denoted by $\emptyset$) and $ G_T^r$ is a complete graph,
i.e. $ G_T = (\emptyset, \emptyset,  G_{\textrm{complete}} )$ where $G_{\textrm{complete}}$ denotes the complete graph.

For a given WSN, assume that  we are given a target graph $ G_T$ that reflects the available prior information.
A natural question is how to incorporate $G_T$ into the classical evaluation metrics, in order to address
 practical concerns (e.g. Observations 1-4). 

First, we briefly review the classical performance metrics in terms of a design graph $G_D$.

%\begin{definition} \label{old_metric} [Classical evaluation metrics]

\begin{itemize}

\item {\bf Direct connectivity coverage}  
is the fraction of  direct links to all possible links in the network, i.e. the probability that a given pair of nodes can directly connect.

\item {\bf Average path length} 
is the expectation of the length of the shortest path between two nodes drawn uniformly from the network.
It can be calculated as the average length of the shortest paths between pairs of nodes in $G_D$.
It is defined to be infinite ($\infty$) if there exist two nodes that cannot establish a connection path.

\item {\bf Network resiliency} 
$NR_x$ measures the fraction of uncompromised external links when $x$ sensor nodes are captured. It can be calculated as the fraction of edges that do not contain keys employed in the compromised nodes.% in the worst case.\footnote{
%The ``worst case'' is with regard to the choice of $x$ nodes that are captured. In other words, 
%$R_x$ should be calculated as the minimal value over all possible combinations of $x$ nodes.} 

\item {\bf Storage overhead} 
measures the memory required to store the keys in each node, often calculated as the size of each block.

\item {\bf Network scalability}  
is the total number of keys needed, for a given number of nodes.
%\footnote{We take this definition because the number of nodes is usually given in a graph-based design. It may also be calculated as the maximal number of generated nodes given a key pool \cite{unital}. }
\end{itemize}
%\end{definition}

Consider a target graph $ G_T$. Some of the above criteria need to be modified accordingly.% in a natural way. 

%\begin{definition} \label{modified_metric} [Graph-based evaluation metrics]

\begin{itemize}
\item {\bf Direct connectivity coverage (DCC) and average path length (APL)} 

If two nodes do not communicate (or are not connected in $ G_T^c  \cup  G_T^r$), whether they share keys should not be 
considered into the evaluation of a KPS (Observation 1 and Observation 3).  
We therefore restrict the calculations of the two metrics
to the edge set $E( G_D) \cap  E( G_T^c  \cup  G_T^r) $, 
e.g. only consider the edges in $ G_D$ that also appear in $ G_T^c  \cup  G_T^r$. 

%\# Put the below metric after the main part, as an introduction to the future work --- weighted graphs 
\item {\bf Direct important connectivity coverage (DICC)} 

%Important connections such as in Observation 2 are represented by the edges of $ G_T^c$. 
Direct important connectivity coverage can be calculated as the direct connectivity coverage restricted 
to $ G_T^c$, i.e. $  | E( G_T^c) \cap E( G_D)  | /  | E( G_T^c) | $ with $|\cdot|$ representing the cardinality of a set.
This metric is meaningful only when  $ G_T^c$ is not empty.

\item {\bf Network resiliency (NR)} 

Observation 4 provides a scenario where certain nodes are required not to communicate.
This is represented by the edges of $ G_T^u$. 
Reflected in the metric of network resiliency, if two compromised nodes are connected in both $ G_D$ and $ G_T^u$,
the common keys they share are regarded as captured. Thus, $NR_x$ can be calculated as 
the fraction of of edges that do not contain keys that are
employed by edges in $E( G_D) \cap E( G_T^u)$ or the $x$ compromised nodes.
It reduces to the classical case when $ G_T^u$ is a null graph.  
%\footnote{For two sets $A,B$, define $A-B$ to be the set of elements that are in $A$ but not in $B$.}
%in the worst case.

%\item {\bf Storage overhead (SO) and network scalability (NS)} 

%The two metrics remain the same as in Definition \ref{old_metric}. 
\end{itemize}

%\end{definition}

%Now we give a broad definition of ``graph-based KPS''.
\begin{definition} \label{def_gdesign} 
A KPS is \textbf{graph-based}, if it is designed based on the target graph. Its performance is evaluated based on the
metrics above.
\end{definition}

%motivation that we separate two cases?
Since it is not easy to provide a universal design that is suitable for any situation, 
we focus on the following two different cases of graph-based KPS.
%The computation of the cases are different.
%In later sections of this work, we are especially interested in the following two cases of graph-based KPS.

\begin{scenario} \label{example1}
Consider the case when $G_T = (\emptyset, \emptyset, G_T^r)$, i.e.  every two nodes may or may not communicate. 
	Notice that in this scenario, we need the modified definition of ``direct connectivity coverage'' and ``average path length''.% due to Definition, while other metrics are not. 
\end{scenario}
	
\begin{scenario} \label{example2}
Consider the case when $ G_T = ( G_T^c, G_T^u, \emptyset )$, i.e. $G_T^c \cup G_T^u$ is a complete graph, and two nodes either must communicate or are required not to communicate.
 	%Notice that in this scenario, the classical ``network resiliency'' is modified, and we need the metric ``direct important connectivity coverage'' in Definition \ref{modified_metric}.
	Notice that in this scenario, we need the modified definition of ``network resiliency'', and the metric ``direct important connectivity coverage''.
\end{scenario}

%=== section ===
 \section{The $G_T = (\emptyset, \emptyset, G_T^r)$ scenario} \label{sec_example_can}

In this scenario, $G_T^r$ contains the information about pairs of nodes that do not need direct connections. 
 %It is the trivial case when $G_T$ is a complete graph, it reduces to the original KPS design.
 If $G_T^r$ is non-trivial, i.e. $G_T^r$ is not a complete graph,  
 we may improve the performance
 %\footnote{Recall the metrics  related to graph-based design are: direct connectivity coverage, average path length, direct important connectivity coverage, network resiliency, storage overhead, network scalability. }
 by employing the extra information provided by $G_T^r$.
 
For comparison, we first consider the trivial case in which $G_T^r$ is a complete graph.

\vspace{0.3cm}

(1) \textit{When $G_T^r$ is a complete graph},
 Bechkit et al. \cite{unital} propose a highly scalable KPS using unital design. This 
was shown to outperform other KPS in many aspects. 
%Since $g$-designs  generalize unital designs, their studies may lead to a more general framework for KPS design.
Here, we examine a more general case. We use a $(\lambda, k, r, v, b)$-BIBD with $\lambda=1$
%\footnote{ Recall that a BIBD with $\lambda=1$ is trivially a $g$-design with $g=1$. }
for KPS design,
%This is a case more general than unital design.
and evaluate the KPS parameters as follows (in terms of $v$ and $k$). 

\vspace{0.3cm}

%----------- evaluation of lambda=1 design ------------------
\begin{itemize}
\item \emph{Direct connectivity coverage / Direct important connectivity coverage}:   
\begin{align} \label{DCC}
DCC &= \frac{bd/2}{b(b-1)/2} = \frac{d}{b-1}  = \frac{\frac{v-k}{k-1} k }{\frac{v(v-1)}{k(k-1)} -1 } \nonumber \\
&=\frac{(v-k) k^2 }{v(v-1)-k(k-1) }.
\end{align}

\item \emph{Average path length}: If two nodes are not connected, there are $u>1$ nodes connecting both of them, so
the minimum path length between these two nodes is equal to two. The average path length is thus  
\begin{equation} \label{APL}
APL = DCC + 2(1-DCC) = 2 - DCC.
\end{equation}

\item \emph{Network resiliency}: For approximate analysis, we assume that the captured nodes are uniformly distributed among all the nodes. Since each key occurs in $r$ blocks among the total number of $b$ blocks, the probability that a key is not compromised when $x$ nodes are captured is 
$
{ b - r \choose x} / {b \choose x},
$ where $b = \frac{v(v-1)}{k(k-1)}$.
Further, the probability that a given link is compromised is 
\begin{equation} \label{NR}
NR_{x }=  \frac{ { b - r \choose x} }{ {b \choose x} }.
\end{equation}

\item \emph{Storage overhead}: 
It is the size of each block, i.e. $SO = k.$

\item \emph{Network scalability}:
It is the total number of keys, i.e. $NS = v.$
\end{itemize}

\vspace{0.3cm}

(2) \textit{When $G_T^r$ is not a complete graph},
in order to exploit the graph information, let us consider a network in which
%we take the following scenario as an example. There are 
there are $s$ groups and each of the group contains $b_0$ sensor nodes. 
For each group, there are $\tau_0$ ($0 \leq \tau_0 \leq b_0/s$) ``central nodes'' (or nodes of higher rank) who are responsible  
for collecting information from all the other nodes in the same group. Besides this, between any two groups only the 
central nodes could establish connections; in other words, a ``non-central node'' can only communicate with the nodes
within the same group. %(see Figure \ref{Obs1} for an illustration).
%\footnote{This scenario may happen when, for example, non-central nodes are not permitted to higher-leveled communications, or they can only communicate within a limited range.  } 
In terms of a target graph, $G_T^r$ is isomorphic to the following matrix $[ J_{mn} ]_{sb_0 \times sb_0}$
(two graphs are isomorphic if the vertices are the same up to a relabeling):
\begin{align}
&J_{mn}  = 1 , \quad \forall 0 \leq (m \textrm{ mod } b_0), \, (n \textrm{ mod } b_0) < \tau_0; \label{line1} \\
&J_{mn} = 1, \quad \forall \left \lfloor \frac{m}{b_0} \right \rfloor = \left \lfloor \frac{n}{b_0} \right \rfloor; \label{line2} \\
&J_{mn} = 0, \quad \textrm{ otherwise. } \label{line3}
\end{align}
Here, $\lfloor x \rfloor$ denotes the largest integer that is no more than the real number $x$.
Equations (\ref{line1}) and (\ref{line2}) represent the possible connections among all the central nodes and among the nodes in each group, respectively. 

Our goal is to design a KPS such that any (central or non-central) node could transfer its information 
to any other node in the network, while satisfying the required performance metrics. 
A possible graph-based KPS design is to assign keys, e.g. via BIBD with $\lambda=1$,
%\footnote{Wilson's theorem states that $(v,k,\lambda)$-BIBD exists for all sufficiently large integers $v$ for which $r$ and $b$ are calculated to be integers in Equation \ref{ddl} \cite{Wilson3}. Thus, the parameters of BIBD with $\lambda=1$ are quite flexible, which validates the following network performance analysis. Direct constructions can be found in \cite{BD,colbourn2010handbook}.} 
for each of the $g$ groups, 
together with the group of all central nodes.
We now evaluate the network performance of this graph-based design, and compare it with the classical 
way. 

Let $v_0$ and $(s+1)v_0$ be respectively the size of the key pool for each group of the graph-based KPS and
for the classical KPS. Let $b = sb_0, v = (s+1)v_0$. 
Then we compute:
%The following evaluations are carried out in terms of $s,b_0,v_0, \tau_0$. 

%----------- evaluation of graph-based design ------------------
\begin{itemize}
\item \emph{Direct connectivity coverage / Direct important connectivity coverage}:   

For classical KPS, from $b = \frac{v(v-1)}{k(k-1)}$ we obtain $k = O(\sqrt{\frac{v(v-1)}{b}}) = O(vb^{-\frac{1}{2}})$,
\footnote{$O(\cdot)$ notation: $f = O(g)$ means there exists a positive constant $c$ such that $c^{-1}g < f < cg$.}
 and from Equation (\ref{DCC}) we  further obtain 
\begin{align}
DCC = O(\frac{vk^2}{v^2}) = O(vb^{-1}) = O(v_0 b_0^{-1}).  \label{DCC_old}
\end{align}
For approximate analysis, 
we assume that the edges of unital design are uniformly distributed in $E(G_T^r)$ and $ E^{'}(G_T^r) $ (the complement
of $E(G_T^r)$). This implies that
$$\frac{ | E(G_T^r) \cap E(G_D)  | }{ | E(G_T^r) | } = \frac{ | E(G_D)  | }{ {b \choose 2} }.$$

For graph-based KPS,  the direct connectivity coverage within one group is
$ O(v_0 b_0^{-1}),$
and within the central nodes is 
$O(v_0 (s\tau_0)^{-1}) \geq O(v_0 b_0^{-1}) $ 
using the same reasoning.
So the overall DCC is at least 
\begin{align}
DCC_G = O(v_0 b_0^{-1}), \label{DCC_new}
%\footnote{Note that the subscript ``G'' denotes metrics of graph-based KPS. }
\end{align}
which is as good as Equation (\ref{DCC_old}).

\item \emph{Average path length}: 
The given target graph requires that any path across two groups consists of two types of connections:
normal node with central node, and central node to central node. So we consider them separately.

For classical KPS, within a group or among central nodes, the average path length is 
%at most (a global srg structure may not be retained locally)
less or equal to
\begin{align}
APL = DCC + 2(1-DCC) = 2 - DCC. \label{APL_old}
\end{align}

For graph-based KPS, connections within any group form a SRG, so the average path length is  
\begin{align}
APL_G  = DCC_G + 2(1-DCC_G) %= 2 - DCC_G 
\approx APL. \label{APL_new}
\end{align}

\item \emph{Network resiliency}: 
We assume that the captured nodes are uniformly distributed among all the nodes. For classical KPS, each key occurs in $r$ blocks (from the total number of $b$ blocks), and thus the probability that a key is not compromised when $x$ nodes are captured is
\begin{align} 
NR_{x}=  \frac{ { b - r \choose x} }{ {b \choose x} } =  \frac{ { b - O(b^{\frac{1}{2}}) \choose x} }{ {b \choose x} }, \label{NR_old}
\end{align}
where we have applied $r = \frac{v-1}{k-1} = O(vk^{-1}) = O(b^{\frac{1}{2}}) $. 
This is also the probability that a given link is not compromised.

For graph-based KPS, 
each key occurs in $r_0 = O(b_0^{\frac{1}{2}})$ blocks (from the total number of $b$ blocks), and thus the probability that a key is not compromised when $x$ nodes are captured is 
$
{ b - r_0 \choose x} / {b \choose x}.
$
Further, the probability that a given link within any group is not compromised is 
\begin{align}
NR_{Gx}=  \frac{ { b - r_0 \choose x} }{ {b \choose x} } 
		= \frac{ { b - s^{-\frac{1}{2}}  O(b^{\frac{1}{2}} ) \choose x} }{ {b \choose x} }; \label{NR_new}
\end{align}
The resiliency for the central nodes $\geq$ $NR_{Gx}$, because $s\tau_0 \leq b_0$. As a result, the resiliency in 
Equation (\ref{NR_new}) is greater than that in (\ref{NR_old}).

\item \emph{Storage overhead}: 
For classical KPS,
storage overhead is given by 
\begin{align} 
SO = k = O(vb^{-\frac{1}{2}}) = s^{\frac{1}{2}} O(v_0 b_0^{-\frac{1}{2}} ) . \label{SO_old}
\end{align}

For graph-based KPS, 
this is equal to $ k_0 = O(v_0 b_0^{-\frac{1}{2}}) $ for a normal node, and $2k_0 $ for a central node, 
so in average:
\begin{align}
SO_G = O(v_0 b_0^{-\frac{1}{2}}) = s^{-\frac{1}{2}} SO .  \label{SO_new}
\end{align}

\end{itemize}

In summary, under the same network scalability,  
the direct connectivity coverage and average path length of graph-based KPS
are no worse than those of the classical ones, while the 
network resiliency and storage overhead are comparatively improved. Thus, the overall performance is improved.
%by virtue of exploiting the information conveyed by the target graph.

%================== section ================== 
\section{The $G_T = (G_T^c, G_T^u, \emptyset)$ scenario} \label{sec_example_should}

In this scenario, we want to have a KPS whose design graph is exactly $G_T^c$. 
We first consider the case where $G_T^c$ is (by coincidence) the design graph of certain $g$-design with $g=2$ and parameters $(\lambda, k, r, v, b)$.
We start with such simple and ``ideal'' case for two main reasons: first, it provides a benchmark for a more general-purposed approach to be introduced in the next section; second, designs for more complex target graphs may be derived based upon the ideal ones. 

%compare -- first way
We immediately obtain a satisfying KPS by the natural mapping between blocks of the $g$-design and key rings. We refer to it as the ``natural'' method.

The storage overhead
%\footnote{Recall that storage overhead can be calculated as the size of each block.} 
is $k$.
The scalability
%\footnote{As discussed in Example \ref{example_cute}, it is measured by the total number of points required.}
is $v$.
As for the network resiliency, 
if one node (block) is captured, there will be ${\lambda \choose 2} {k \choose 2}$ connections compromised. 
To observe this, we first notice 
that any pair of points in the block appear in $\lambda$ blocks, and there are ${k \choose 2}$ such pairs,
so ${\lambda \choose 2} {k \choose 2}$ connections are compromised. Moreover, any other connection is secure,
since the corresponding two blocks share at least one key that is not captured.

%=== example ===
\begin{example} \label{example_cute}

Let $G_T^c$ be the graph
shown in Figure \ref{14pointGraph}(a). As we can see, the nodes of $G_T^c$ could be grouped into seven disconnected pairs
with all other edges connected in the graph. For clarity, the complement of $G_T^c$ is 
given in Figure \ref{14pointGraph}(b). 
Assume we would like to construct a KPS whose design graph is $G_T^c$. 

\begin{figure}[h]%[!t]
\centering
\subfloat[]{\includegraphics[width=1.5in]{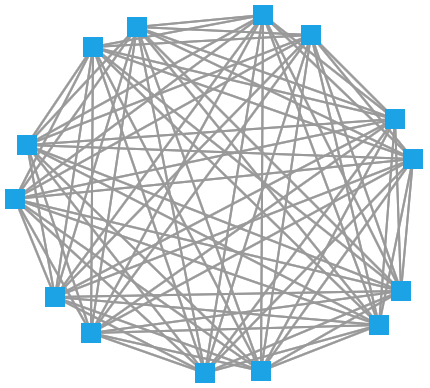}%
\label{14point}}
\hfil
\subfloat[]{\includegraphics[width=1.6in]{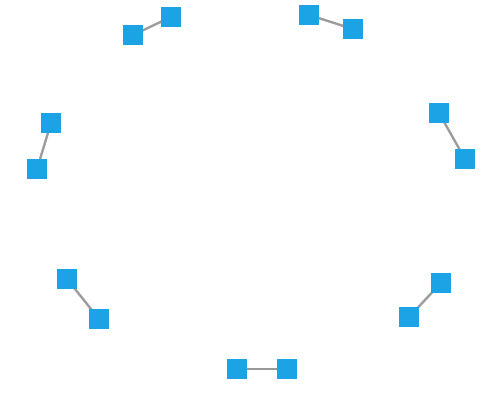}%
\label{14pointC}}
\caption{Graph $G_T^c$ and its complement}
\label{14pointGraph}
\end{figure}

%compare -- first way
We construct the following $g$-design with $g=2$ and parameters 
$(\lambda, k, r, v, b) = (3, 4, 7, 8, 14)$:
%\footnote{This $g$-design example will also appear in Section \ref{uniqueExample2}.}
%We denote the $8$ points by 
\begin{align} 
V =& \{a_1 a_2 a_3 a_4 a_5 a_6 a_7 a_8\} , \label{Stanton_V} \\
%\end{align}
%and let
%\begin{align} 
\mathfrak{B} =& 
\{\{a_1 a_2 a_3 a_4\},\{ a_1 a_2 a_5 a_6 \},\{ a_1 a_2  a_7 a_8 \},\{ a_1 a_2  a_6 a_8 \}, \nonumber \\
&\{ a_1 a_2  a_5 a_7 \},\{ a_1 a_4  a_5 a_8 \},\{ a_1 a_4  a_6 a_7 \}, \nonumber \\
%and their complements
&\{ a_5 a_6 a_7 a_8\},\{ a_3 a_4  a_7 a_8 \},\{ a_3 a_4  a_5 a_6 \},\{ a_2 a_4  a_5 a_7 \}, \nonumber \\
&\{ a_2 a_4  a_6 a_8 \},\{ a_2 a_3  a_6 a_7 \},\{ a_2 a_3  a_5 a_8 \} \}. \label{Stanton_B}
\end{align}

Obviously ($V, \mathfrak{B}$) forms a $g$-design ($g=2$) whose design graph is $G_T^c$.
In other words, we have obtained a satisfying KPS with $a_i$, $i=1,\cdots,8$, representing the keys. 

Next we evaluate the performance of this construction by computing resiliency, storage overhead,  and scalability measures.

We consider the simple case when one node is captured, say $V_1 = \{a_1 a_2 a_3 a_4\}$. In this case, there are $18$ connections compromised. To observe this,  we first notice 
that any two points of $V_1$ appear in exactly $3$ blocks; thus, there are ${3 \choose 2} {4 \choose 2} =18$ compromised connections in total. Besides this, if 
%the two points that two blocks intersect do not belong to $V_1$, say $\{a_1 a_5\}$, the two blocks are able to communicate in a secure way. 
both points of the intersection of two blocks do not belong to $V_1$, the two blocks are able to communicate in a secure way. 
%For example, they may use key $\{a_5\}$ that is not compromised, or use the key that is a concatenation of $\{a_1,a_5\}.
%However, if two nodes are captured, the worst case is that \{abcd\} and \{efgh\} are captured. If so, 
%all the connections are compromised. 

The storage overhead is the size of each block, i.e. $k=4$.
The network scalability
%notice that the number of blocks is fixed (determined by $|V(\bm G_T^c)|$), and we may measure it by 
is the total number of points, i.e. $v=8$.
\end{example}

Now the question is: 
for any given target graph $G_T = (G_T^c, G_T^u, \emptyset)$, does there exist a KPS whose design graph $G_D$ is exactly $G_T^c$? 
In fact, we have a positive answer that is guaranteed by the following algorithm.

\section{Matching and Reducing Algorithm (MAR)} \label{sectionMAR}

A schematic diagram of the Matching and Reducing Algorithm (MAR) is depicted in Table~\ref{MAR}.

\begin{table}[t]
\renewcommand{\arraystretch}{1.4}
\renewcommand{\arraystretch}{1.2}
\caption{Matching and Reducing Algorithm}
\begin{center}
\begin{tabular}{l}
\toprule[1pt] {\bf Input:}\label{MAR}

\hspace{0.5em} $G_T^c=[V(G_T^c), E(G_T^c)]$,  positive integer $c_0$. \\
\hspace{3.5em} Let $b=|V(G_T)|, e=|E(G_T)|$. \\
{\bf Initialization:} \\
\hspace{1.5em} $l=0,  G^l = G_T^c$. \\
\hspace{1.5em} Generate $e$ different keys $K = \{k_1,\cdots,k_e\}$.\\
\hspace{1.5em} Assign $K$ to $E(G^0)$, s.t. each edge is assigned a unique key.\\
\hspace{1.5em} Define $\mathfrak{B} = \{ B_1,\cdots,B_b \}$, where\\
\hspace{1.5em} $B_n = \{k_t \mid k_t \textrm{ is assigned to an edge which is incident to node } n \}$.\\

{\bf Repeat} (Clique reduction procedure)\\
\hspace{1.5em}$l=l+1$;\\
\hspace{1.5em}Find a clique $ C^l$ in $ G^{l-1}$ whose size is no larger than $c_0$;\\
\hspace{1.5em}Denote $K_{ G^{l-1}}$ as the set of keys assigned to $E( C^l)$;\\
\hspace{1.5em}Update the blocks:\\
\hspace{3.5em}$B_m = B_m - K_{ G^{l-1}}, \forall m \in V( C^l)$;\\
\hspace{1.5em}Arbitrarily choose a key $k^l$ from $K_{ G^{l-1}}$:\\
\hspace{3.5em}$B_m = B_m \cup \{k^l\}, \forall m \in V( C^l);$\\
\hspace{1.5em}Update from $ G^{l-1}$ to $ G^{l}$:\\
\hspace{3.5em}$E( G^l) = E( G^{l-1}) - E( C^l)$;\\

{\bf Until} 
\hspace{0.5em} no clique of size greater than $2$ can be found.\\

{\bf Output: } 
\hspace{0.5em} $\mathfrak{B}, V=\cup_{m=1}^b B_m.$\\
\toprule[1pt] 
\end{tabular}
\end{center}
\end{table}

In the initialization step, a unique key is assigned to each edge, i.e.  two nodes
of that edge contain the key.  Thus, the key ring of each node has been determined. 
However, we are motived to reduce the size of the key pool and key rings. Notice that 
for any clique $ C$ in $ G_T^c$, it will not change the design graph if the distinct keys
assigned to $E( C )$ are replaced by only one key; in other words, if all the nodes of $ C$ share a common key,
they still directly connect with one another, and edges outside $ C$ are not affected. 
Therefore,  MAR looks for cliques in the current target graph, and reduces the number of 
keys within each clique to one, and then updates the target graph by removing the clique. 
The size of the clique, however, is upper-bounded by a constant integer $c_0$ to ensure a reasonable network resiliency.
Consider, for example, if $ G_T^c$ is a complete graph, one key is enough for full connectivity; however, the 
the whole network is compromised as long as any one node is captured. 

\begin{figure}[h]%[!t]
\centering
\includegraphics[width=2.5in]{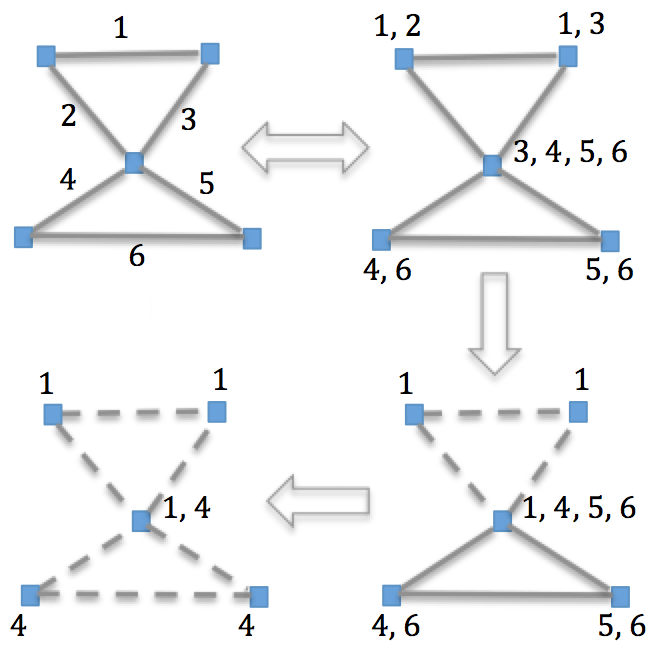}%
\caption{An illustrating example where MAR with $c_0=3$ is applied to a simple graph}
\label{MAR_example}
\end{figure}

An illustrating example is given in Figure \ref{MAR_example}.

We elaborate on the other aspects of this algorithm:
\begin{enumerate}
\item We did not give the details of how to detect cliques, or the optimality criteria for clique selection. 
Indeed, MAR encompasses many more algorithms. For example, one may
 propose a specific algorithm that deviates from MAR, in order to minimize the total number of keys ($\min |V|$), etc.
%\footnote{The new algorithm is proposed to minimize a specific objective function, so it should be equipped with more delicate procedures in choosing cliques. }  
%But we leave this for future work here. 

\item It is interesting that MAR is applied as a technical tool to the proof of Theorem \ref{bridge_s1} (in Appendix \ref{proof_bridge_s1}).

\item The upper bound of the clique size in MAR is designed to ensure network resiliency $NR_x$. 
But is there any theoretical lower bound of $NR_x$ if the clique size is bounded by $c_0$?  
The following theorem gives a positive answer.   
\end{enumerate}

\subsection{Network Resiliency for MAR}

\begin{theorem} \label{Rx_bound} 
For a given graph $ G_T^c$, assume that the degree of any node is no larger than $d$, the 
network resiliency $NR_x$ of the KPS determined by the Matching and Reducing Algorithm satisfies 
\begin{equation} \label{Rx_LB}
 NR_x \geq 1 -  \frac{ x {c_0 \choose 2} \lceil \frac{d}{c_0-1} \rceil}{ |E( G_T^c)| }, 
\end{equation}
where $\lceil \frac{d}{c_0-1}  \rceil$ denotes the smallest integer that is no less than $\frac{d}{c_0-1} $.
\end{theorem}

\begin{proof}

Let us consider an arbitrary node $x$. Each key that is assigned to $x$ is associated with one clique in $G_T^c$, due to the clique reduction procedure of the MAR algorithm (an edge can be regarded as a clique of size $2$). Assume that $x$ is associated with $M$ cliques, and that clique $m$ is of size $\lambda_m, m=1,\cdots,M$.  
 Let $y_m = \lambda_m - 1$. It is clear that
\begin{equation} \label{bound_constraint}
\sum_{m=1}^M y_m \leq d, \quad 1 \leq y_m \leq c_0 - 1. \quad (\textrm{$c_0$ is given in the algorithm})
\end{equation}
%
%Assume that the MAR algorithm is used in the design. 
Moreover, the number of compromised links when node $x$ is captured is:
$$ F(y_1, \cdots, y_m) = \sum_{m=1}^M {y_m + 1 \choose 2}.$$
We now evaluate the maximum for $F$ under constraint (\ref{bound_constraint}). 

Let $f(y) = {y + 1 \choose 2}$. Given positive numbers $a,b,s$ in such a way that $s-a >0, s-b >0$,
 it is easy to observe that $f(a)+f(s-a) < f(b) + f(s-b)$
if and only if $|a-\frac{s}{2}| < |b - \frac{s}{2}|$.

Given two positive variables $y_1$ and $y_2 $  satisfying 
$y_1 + y_2\leq c_0 - 1$, we have
$$
| y_1 - \frac{y_1 + y_2}{2} | < | (y_1 + y_2) - \frac{y_1 + y_2}{2}|.
$$
%$F(\cdot)$ is increased if $y_1, y_2$ are combined to be $y_1 + y_2$, i.e. 
We conclude that 
\begin{equation} \label{new_Ineq1}
F(y_1, y_2, \cdots,  y_m) < F(y_1 + y_2, \cdots, y_m).
\end{equation}

Moreover, given two positive variables $y_1$ and $y_2$ satisfying
$y_1 + y_2> c_0 - 1, \,  y_1 \leq y_2< c_0 - 1$, we have 
$$
| y_2- \frac{y_1 + y_2}{2} | < | c_0 - 1 - \frac{y_1 + y_2}{2}|.
$$
%$F(\cdot)$ is increased if $y_1 , y_2$ are replaced by $c_0 -1,  y_1 + y_2- (c_0-1) $, i.e.
We conclude that 
\begin{equation} \label{new_Ineq2}
F(y_1, y_2, \cdots,  y_m) < F(c_0 -1,  y_1 + y_2- (c_0-1), \cdots, y_m).
\end{equation}

By continuous application of Inequalities (\ref{new_Ineq1}) and (\ref{new_Ineq2}), $F$ is maximized when 
\begin{align*}
&y_1 = \cdots = y_{M-1} = c_0 -1, \quad y_M = d - (M-1)(c_0-1), \\
&M = \left \lceil \frac{d}{c_0-1} \right \rceil .
\end{align*}
Thus, 
\begin{equation} \label{Rx_singleLB}
F(\cdot) \leq {c_0 \choose 2} M = {c_0 \choose 2} \left \lceil \frac{d}{c_0-1} \right \rceil  = F_0.
\end{equation}

If $x$ nodes are captured, the worst case is that they do not share keys and $xF_0$
connections are compromised, which implies the result in (\ref{Rx_LB}). 

\end{proof}

\subsection{Example}
%compare -- second way
In this section, we revisit the special case discussed in Section \ref{sec_example_should}, i.e. 
$G_T^c$ in $G_T = (G_T^c, G_T^u, \emptyset)$ is the design graph of a $g$-design with $g=2$ and parameters $(\lambda, k, r, v, b)$.
Now we apply the Matching and Reducing Algorithm and choose the parameter $c_0$ to be $\lambda$, so that the network resiliency is not worse than the ``natural'' method. Here are the reasons:
\begin{itemize}
\item Due to Theorem \ref{bridge_s2}, the edge set of $G_T^c$ is a disjoint union of ${v \choose 2}$ 
subsets, each of which forms a clique of size $\lambda$. Further, when MAR with $c_0=\lambda$ is applied,  
it is clear that the minimal number of keys is achieved when the $ C^l$ in each step is one of the ${v \choose 2}$ cliques.
  
\item Therefore, there are ${v \choose 2}$ keys in total, and the number of keys required by each node is 
$$
 \frac{d}{\lambda-1} = \frac{ 2 | E(G_T^c) |}{b} \frac{1}{\lambda-1}
= \frac{ 2 {v \choose 2} {\lambda \choose 2} }{b} \frac{1}{\lambda-1} = {k \choose 2},
$$ 
where the last equality is due to Equation (\ref{bridge2_par}).

\item If one node is captured, there will be ${k \choose 2} {\lambda \choose 2}$ connections compromised, which is the same as the ``natural'' method.

\item Finally, if $c_0$ is chosen to be larger, the network resiliency obviously decreases.
%\footnote{A concrete example has been given in Example \ref{example_cute}, and we do not elaborate here.} 

\end{itemize}
%In summary, to keep the network resiliency not worse than the first method, the second method requires storage overhead  $d/(\lambda-1) = {k \choose 2}$ and network scalability ${v \choose 2}$ compared to respectively $k$ and $v$ in the first method. 

%\vspace{0.2cm}
%compare -- second way
Returning to Example \ref{example_cute}, the parameter $c_0$ is chosen to be $3$.
There are ${v \choose 2} =28$  keys in total, and each node requires $ {k \choose 2}  = 6$ keys.
The keys/key rings can be realized as:
%
%\begin{align*}
%\mathfrak{B} = &
%\{ \{ a_1 a_2 ,  a_1 a_3 ,  a_1 a_4 ,  a_2 a_3 ,  a_2 a_4 ,  a_3 a_4 \}, 
%\{ a_1 a_2 ,  a_1 a_5 ,  a_1   a_6 ,  a_2  a_5  ,  a_2  a_6 ,  a_5 a_6 \},
%\{ a_1 a_2 ,  a_1  a_7,  a_1   a_8,  a_2  a_7,  a_2  a_8,  a_7 a_8 \}, 	\\
%&\{ a_1 a_3 ,  a_1   a_6 ,  a_1   a_8, a_3 a_6 ,  a_3 a_8,  a_6 a_8 \},
%\{ a_1 a_3 ,  a_1 a_5 ,  a_1   a_7,  a_3 a_5  ,  a_3 a_7,  a_5 a_7 \},
%\{ a_1 a_4 ,  a_1 a_5 ,  a_1   a_8,  a_4 a_5  ,  a_4 a_8,  a_5 a_8 \},	\\
%&\{ a_1 a_4 ,  a_1   a_6 ,  a_1   a_7,  a_4 a_6 ,  a_4 a_7,  a_6 a_7 \},
%%and their complements
%\{ a_5 a_6 ,  a_5 a_7 ,  a_5 a_8 ,  a_6 a_7 ,  a_6 a_8 ,  a_7 a_8 \},
%\{ a_3 a_4 ,  a_3 a_7,  a_3 a_8,  a_4 a_7,  a_4 a_8,  a_7 a_8 \},	\\
%&\{ a_3 a_4 ,  a_3 a_5  ,  a_3 a_6 ,  a_4 a_5  ,  a_4 a_6 ,  a_5 a_6 \},
%\{ a_2 a_4 ,  a_2  a_5  ,  a_2  a_7,  a_4 a_5  ,  a_4 a_7,  a_5 a_7 \},
%\{ a_2 a_4 ,  a_2  a_6 ,  a_2  a_8,  a_4 a_6 ,  a_4 a_8,  a_6 a_8 \}, 	\\
%&\{ a_2 a_3 ,  a_2  a_6 ,  a_2  a_7,  a_3 a_6 ,  a_3 a_7,  a_6 a_7 \},
%\{ a_2 a_3 ,  a_2  a_5  ,  a_2  a_8,  a_3 a_5  ,  a_3 a_8,  a_5 a_8 \} \}.
%\end{align*}
\begin{align*}
\mathfrak{B} = &
\{ \{ 12 ,  13 ,  14 ,  23 ,  24 ,  34 \}, 
\{ 12 ,  15 ,  16 ,  25  ,  26 ,  56 \}, \\
&\{ 12 ,  17,  18,  27,  28,  78 \}, 	
\{ 13 ,  16,  18, 36,  38,  68 \}, \\
&\{ 13 ,  15,  17,  35  ,  37,  57 \},
\{ 14,  15, 18,  45 ,  48,  58\},	\\
&\{ 14,  16,  17,  46,  47,  67\},
%and their complements
\{ 56,  57 ,  58,  67,  68 ,  78 \},\\
&\{ 34,  37,  38,  47,  48, 78 \},	
\{ 34,  35,  36,  45 ,  46,  56 \},\\
&\{ 24,  25 ,  27,  45 ,  47,  57 \},
\{ 24,  26,  28,  46 ,  48,  68 \}, 	\\
&\{ 23,  26 ,  27,  36,  37,  67 \},
\{ 23 ,  25 ,  28,  35 ,  38,  58 \} \}.
\end{align*}
Here, each key is uniquely denoted by a two-digit integer. For example, $12$ represents a key, and $13$ represents another key.

Clearly, if one node is captured, $6 \cdot { c_0 \choose 2} = 18 $ connections are compromised, which is the same as the ``natural'' approach.

 Moreover, if $c_0$ is chosen to be $2$, the network resiliency is improved. This is because
 every connection is secured by a unique key, and if one node is captured, only $d=12<18$ connections are compromised. 
 However, the storage overhead increases and network scalability decreases. 
 
 If $c_0$ is chosen to be $4$, the network resiliency decreases. To observe this, consider the following case:

Label the the $14$ nodes to be $n_1, \cdots, n_{14}$, and let two nodes $n_i, n_j$ be disconnected if and only if 
$|i - j|=7$ 
(Figure \ref{14pointGraph}). If we apply MAR with $c_0=4$, then it can be assumed that the following four cliques of size $4$ appear 
in the ``clique reduction procedure'' (by possible relabeling):
\begin{align*}
&\{n_1, n_2, n_3, n_4\}, \{n_1, n_5, n_6, n_7\}, \\
& \{n_1, n_9, n_{10}, n_{11}\}, \{n_1, n_{12}, n_{13}, n_{14}\}.
\end{align*}
This means that if node $n_1$ is captured, the four keys, along with the $4 \cdot {4 \choose 2} = 24>18$ 
connections among the above four cliques, are compromised. 

%\vspace{0.1cm}

In summary, the Matching and Reducing Algorithm provides a general solution for KPS design given an arbitrary target graph. 
Nevertheless, the previous example reveals that for specific target graphs ($G_T^c$), there is a potential advantage of using $g$-designs ($g>1$) based KPS %over MAR, 
in terms of storage overhead and scalability. 
We believe that $g$-design is a promising design tool for KPS and leave that for future work here.

%% --- to be commented and to be added as numerical simulation !! ----
%In the case of unital design ($v=m^3+1, k=m+1$), the above results reduce to:\footnote{
%The results are consistent with that in \cite{unital}.}
%\begin{align*}
%&DCC_0 = \frac{(m+1)^2 }{m^3+m+1}, \,
%APL_0 = \frac{2m^3-m^2+1}{m^3+m+1}, \\
%&NR_{0x} =  \frac{ { m^3(m-1)  \choose x} }{ {m^2(m^2-m+1) \choose x} },
%SO_0 = m+1,\,
%NS_0 = m^3+1.
%\end{align*}
%On the other side, if the number of nodes $b$ is given, the network performance of unital-based KPS is:
%\begin{align}
%&DCC_0 = O(b^{-\frac{1}{4}}), \, APL_0 < 2, \label{unital_order1} \\
%&NR_{0x} =  \frac{ { b-O(b^{\frac{1}{2}})  \choose x} }{ {b \choose x} }, \,
%SO_0 = O(b^{\frac{1}{4}}), \, NS_0 = O(b^{ \frac{3}{4} }). \label{unital_order2}
%\end{align}
% %-------------

% === section ===
\section{Conclusion} \label{sec_conclusion}

We proposed the concept of graph-based KPS and relevant evaluation metrics, motivated by several practical observations.
We introduced the $g$-designs,  studied some of their connections with graph theory, and applied them to KPS constructions.  
Two specific target graphs are considered. Especially, we introduced an algorithm framework called the Matching and Reducing Algorithm. Examples are provided to demonstrate the performance of the proposed scheme.

%Starting from a specific scenario we proposed a general algorithm called MAR.
%As the graph-based KPS is a new concept, and this paper focuses on the theoretical framework and promising technique tools, we hope to see more development in the future, including more practical algorithms. 

%============================================================ Appendices =====================
\appendices

% --- Appendix ----
\section{Proof of Theorem \ref{bridge_s2} } \label{proof_bridge_s2}

Consider a $g$-design $(V, \mathfrak{B}) $ with $g=2$ and design graph $ G$. Equation (\ref{ddl}) implies that $G$ has $b=\frac{\lambda v(v-1)}{k (k-1)}$ nodes.
Because two connected blocks share exactly $g=2$ keys, all edges in $ G$ are induced by a pair of keys in $V$.
Besides this, any pair of keys induces a clique of size $\lambda$ in $ G$. Two cliques do not share an edge, otherwise there exists two blocks
that share at least three keys. Finally, every block intersect with other  $\frac{k(k-1)}{2} (\lambda-1)$ blocks, because each block contains $\frac{k(k-1)}{2}$ pairs and each pair belongs to $\lambda-1$ other blocks. This implies that $ G$ is regular.

% --- Appendix ----
\section{Proof of Theorem \ref{bridge_s1} } \label{proof_bridge_s1}

Sufficiency:

Assume that a ($\lambda=1, k, v)$-BIBD exists. Let $ G$ be its design graph. % $\bm G$ has $b=\frac{(v-1)v}{(k-1)k}$ nodes. 
Any point in $V$ induces $r$ nodes that are connected to one another, forming 
a clique of size $r$. Besides this, any two cliques induced by two different points do not share an edge, because otherwise 
these two points appear in two different blocks contradicting $\lambda=1$.
Finally, every block intersects with other $(r-1)k$ blocks, because each block contains $k$ points and each point belongs to $r-1$ other blocks. This implies that $ G$ is regular. 

Necessity:

Let $$k=\frac{rv}{b} = \frac{v-1}{r} + 1.$$
Applying MAR to the given graph $ G$ for $v$ iterations, 
we obtain $(V, \mathfrak{B})$ together with an empty graph $ G^l$. 
Due to MAR, $|V| = v$, each key appears in $r$ blocks, and each block in $\mathfrak{B}$ contains 
$$
\frac{d}{r-1} = \frac{2e} {b} \frac{1}{r-1} = 2  \frac{vr(r-1)}{2} \frac{1}{b(r-1)} = k
$$
keys.
Next we prove that any pair of keys appear in exactly one block. To this end, we let $b_{ij}$ be the number of blocks which  
keys $i$ and $j$ both belong to, and count the value of $B = \sum_{1 \leq i,j \leq v} b_{ij}$ in two different ways:
on one hand, because each block contains $k$ points and there are $b$ blocks,
$B$ can be calculated as 
$$
B = \frac{k(k-1)}{2} b = \frac{k(k-1)}{2} \frac{v(v-1)}{k(k-1)} = \frac{v(v-1)}{2}.
$$
On the other hand,
any pair of keys appear in no more than one block, because $G$ is decomposed into $v$ cliques any two of which share no edges. 
%the total number of pairs appearing in the same block 
Moreover, $B$ is no more than the total number of pairs $v(v-1)/2$. 
Therefore, any pair must exactly appear in one block.

\ifCLASSOPTIONcaptionsoff
  \newpage
\fi

\bibliographystyle{IEEEtran}
\bibliography{KPS}

\begin{IEEEbiographynophoto}{Jie Ding}
is a Ph.D. candidate in the School of Engineering and Applied Sciences, Harvard University. 
His current research are in cyclic difference sets, sequence 
designs, time series, and information theory.  
\end{IEEEbiographynophoto}

\begin{IEEEbiographynophoto}{Abdelmadjid Bouabdallah}
received the Engineering degree from USTHB-Algeria, Master degree in 1988 and Ph.D. from university of Paris-sud Orsay (France) in 1991. From 1992 to 1996, he was Assistant Professor at university of Evry-Val-d'Essonne (France) and since 1996 he is Professor at University of Technology of Compiegne (UTC) where he is leading the Networking \& Security research group and the Interaction \& Cooperation research of the Excellence Research Center LABEX MS2T. His research Interest includes Internet QoS,
security, unicast/multicast communication, Wireless Sensor Networks, and fault tolerance in wired/wireless networks. He conducted several large scale research projects founded by well known companies (Motorola Labs., Orange Labs., CEA, etc) as well as academy (ANR-RNRT, CNRS, ANR-Carnot).
\end{IEEEbiographynophoto}

\begin{IEEEbiographynophoto}{Vahid Tarokh}
 is a professor of applied mathematics in the School of Engineering and Applied Sciences, Harvard University. 
His current research interests are in data analysis, network security, optical surveillance, and radar theory.
\end{IEEEbiographynophoto}

\end{document}